\def\confversion{0}
\def\ifconf{\ifnum\confversion=1}
\def\ifnotconf{\ifnum\confversion=0}

\documentclass[11 pt]{article}
\def\showauthornotes{1}
\def\showkeys{0}
\def\showdraftbox{0}
\def\confversion{0}
\def\widemargin{1}
\def\ipadcompile{0}

\usepackage{xspace,enumerate}
\usepackage{amsmath,amssymb}
\usepackage{amsthm}
\ifnum\confversion=0
	\usepackage[toc,page]{appendix}
\fi
\usepackage[T1]{fontenc}
\usepackage[utf8]{inputenc}

\usepackage{thmtools}
\usepackage{thm-restate}
\usepackage{color,graphicx}
\usepackage{boxedminipage}
\usepackage{makecell}
\usepackage{tabularx}
\usepackage{csquotes}

\ifnum\widemargin=0
\usepackage[top=1in, bottom=1in, left=1in, right=1in]{geometry}
\else
\usepackage[top=1in, bottom=1in, left=1.25in, right=1.25in]{geometry}
\fi

\ifnum\showkeys=1
\usepackage[color]{showkeys}
\fi

\definecolor{darkred}{rgb}{0.5,0,0}
\definecolor{darkgreen}{rgb}{0,0.35,0}
\definecolor{darkblue}{rgb}{0,0,0.55}

\usepackage[dvipsnames]{xcolor}

\usepackage[pdfstartview=FitH,pdfpagemode=UseNone,colorlinks,linkcolor=NavyBlue,filecolor=blue,citecolor=OliveGreen,urlcolor=NavyBlue,pagebackref]{hyperref}

\usepackage{enumitem}
\usepackage{comment}
\usepackage{todonotes}

\usepackage[capitalise,nameinlink]{cleveref}

\usepackage{microtype}

\usepackage{mathtools,dsfont}

\usepackage{palatino}
\usepackage[scaled=.95]{helvet}

\ifnum\ipadcompile=1
\usepackage{eulervm}
\else
\usepackage{eulerpx}
\fi

\DeclareFontEncoding{LGR}{}{}
\DeclareSymbolFont{sfitgreek}{LGR}{cmss}{m}{it}
\SetSymbolFont{sfitgreek}{bold}{LGR}{cmss}{bx}{it}
\DeclareMathSymbol{\sfpi}{\mathord}{sfitgreek}{`p}

\usepackage{tcolorbox}

\DeclareMathAlphabet{\mathpazocal}{OMS}{zplm}{m}{n}
\DeclareRobustCommand*{\mathcal}[1]{\mathpazocal{#1}}
\ifnum\showauthornotes=1
\newcommand{\Authornote}[3]{{\sf\small\color{#3}{[#1: #2]}}}
\newcommand{\Authorcomment}[2]{{\sf \small\color{gray}{[#1: #2]}}}
\newcommand{\Authorfnote}[2]{\footnote{\color{red}{#1: #2}}}
\else
\newcommand{\Authornote}[3]{}
\newcommand{\Authorcomment}[2]{}
\newcommand{\Authorfnote}[2]{}
\fi

\ifnum\showdraftbox=1
\newcommand{\draftbox}{\begin{center}
  \fbox{%
    \begin{minipage}{2in}%
      \begin{center}%
        \begin{Large}%
          \textsc{Working Draft}%
        \end{Large}\\
        Please do not distribute%
      \end{center}%
    \end{minipage}%
  }%
\end{center}
\vspace{0.2cm}}
\else
\newcommand{\draftbox}{}
\fi

\declaretheorem[numberwithin=section]{theorem}
\declaretheorem[sibling=theorem]{lemma}
\declaretheorem[sibling=theorem]{claim}

\declaretheorem[sibling=theorem]{fact}
\declaretheorem[sibling=theorem]{corollary}

\theoremstyle{definition}
\declaretheorem[sibling=theorem]{definition}
\declaretheorem[sibling=theorem]{remark}

\declaretheorem[sibling=theorem]{observation}

\newtheorem{algo}[theorem]{Algorithm}

\def\FullBox{\hbox{\vrule width 6pt height 6pt depth 0pt}}

\def\qed{\ifmmode\qquad\FullBox\else{\unskip\nobreak\hfil
\penalty50\hskip1em\null\nobreak\hfil\FullBox
\parfillskip=0pt\finalhyphendemerits=0\endgraf}\fi}

\def\qedsketch{\ifmmode\Box\else{\unskip\nobreak\hfil
\penalty50\hskip1em\null\nobreak\hfil$\Box$
\parfillskip=0pt\finalhyphendemerits=0\endgraf}\fi}

\def\matr#1{\mathsf{#1}}
\newcommand{\HS}{\mathrm{HS}}
\newcommand{\CS}{Cauchy--Schwarz}
\newcommand{\Ut}{\mathrm{U}^2}
\newcommand{\Mt}{\mathrm{M}_t}
\newcommand{\Utt}{\U_t}

\newcommand{\reg}{\mathrm{reg}}

\newcommand{\triv}{\mathrm{triv}}

\newcommand{\ep}{\varepsilon}

\makeatletter
\newcommand{\stackalign}[1]{
	\vcenter{
		\Let@ \restore@math@cr \default@tag
		\baselineskip\fontdimen10 \scriptfont\tw@
		\advance\baselineskip\fontdimen12 \scriptfont\tw@
		\lineskip\thr@@\fontdimen8 \scriptfont\thr@@
		\lineskiplimit\lineskip
		\ialign{\hfil$\m@th\scriptstyle##$&$\m@th\scriptstyle{}##$\crcr
			#1\crcr
		}
	}
}

\let\latexcirc=\circ
\newcommand{\ccirc}{\mathbin{\mathchoice
  {\xcirc\scriptstyle}
  {\xcirc\scriptstyle}
  {\xcirc\scriptscriptstyle}
  {\xcirc\scriptscriptstyle}
}}
\newcommand{\xcirc}[1]{\vcenter{\hbox{$#1\latexcirc$}}}\let\circ\ccirc

\def\to{\rightarrow}
\def\eps{\varexpan}
\def\epsilon{\varepsilon}

\def\eps{\epsilon}

\def\phi{\varphi}

\newcommand{\ol}{\overline}

\newcommand{\ie}{i.e.,\xspace}

\newcommand{\mper}{\,.}

\newcommand{\E}{{\mathbb E}}
\newcommand{\C}{{\mathbb C}}

\newcommand{\Z}{{\mathbb Z}}

\newcommand{\U}{{\mathbb U}}

\DeclarePairedDelimiter\parens{\lparen}{\rparen}
\DeclarePairedDelimiter\abs{\lvert}{\rvert}
\DeclarePairedDelimiter\norm{\lVert}{\rVert}

\DeclarePairedDelimiter\braces{\lbrace}{\rbrace}
\DeclarePairedDelimiter\brackets{\lbrack}{\rbrack}
\DeclarePairedDelimiter\angles{\langle}{\rangle}
\DeclarePairedDelimiterXPP\lnorm[1]{}\lVert\rVert{_2}{#1}

\DeclareMathDelimiter{\given}
      {\mathbin}{symbols}{"6A}{largesymbols}{"0C}

\newcommand{\prob}{\mathsf{Pr}}
\newcommand{\Esymb}{\mathbb{E}}

\newcommand{\Psymb}{\mathrm{Pr}}

\DeclarePairedDelimiterXPP{\Prob}[1]
 {\prob}{\lparen}{\rparen}{}
 {\renewcommand{\given}{\;\delimsize\vert\nonscript\;\mathopen{}}#1}

\makeatletter
\def\Pr#1{%
    \ProbabilityRender{\Psymb}{#1}%
}

\def\Ex#1{%
    \ProbabilityRender{\Esymb}{#1}%
}

\def\condPE#1#2{%
	\@ifnextchar\bgroup
	{\ConditionalProbabilityRender{\widetilde{\Esymb}}{#1}{#2}}
	{\ProbabilityRender{\widetilde{\Esymb}}{#1 \given #2}}
}

\def\ConditionalProbabilityRender#1#2#3#4{
	\renderwithdist{#1}{#2}{#3 \given #4}	
}

\def\ProbabilityRender#1#2{
  \@ifnextchar\bgroup%
  {\renderwithdist{#1}{#2}}
   {\singlervrender{#1}{#2}}
}
\def\singlervrender#1#2{%
   {\mathchoice
       {{#1}\brackets*{#2}}
       {{#1}[ #2 ]}
       {{#1}[ #2 ]}
       {{#1}[ #2 ]}
   }
}
\def\renderwithdist#1#2#3{%
   \@ifnextchar\bgroup
   {\superfancyrender{#1}{#2}{#3}}
   {\mathchoice
      {\underset{#2}{#1}\brackets*{#3}}
      {{#1}_{#2}[ #3 ]}
      {{#1}_{#2}[ #3 ]}
      {{#1}_{#2}[ #3 ]}
     }
}
\def\superfancyrender#1#2#3#4#5{
   \ensuremath{\mathchoice
      {\underset{#1}{{#1}}\left#4 #3 \right#5}
      {{#1}_{#2}#4 #3 #5}
      {{#1}_{#2}#4 #3 #5}
      {{#1}_{#2}#4 #3 #5}
   }
}
\makeatother

 \newcommand\SetSymbol[1][]{%
     \nonscript\:#1\vert
     \allowbreak
     \nonscript\:
     \mathopen{}}
  \DeclarePairedDelimiterX\Set[1]\{\}{%
     \renewcommand\given{\SetSymbol[\delimsize]}
     #1
}

\newcommand{\ip}[2]{\angles*{#1 , #2}}

 \newcommand{\set}[1]{\braces*{#1}}

\newcommand{\opnorm}[1]{\norm*{#1}_{\mathrm{op}}}
\newcommand{\hsnorm}[1]{\norm*{#1}_{\mathrm{HS}}}

\newcommand{\poly}{{\mathrm{poly}}}

\newcommand{\tr}{\mathrm{tr}}

\DeclareMathOperator{\Cay}{Cay}

\DeclareMathOperator{\Tr}{Tr}

\definecolor{cadmiumgreen}{rgb}{0.0, 0.42, 0.24}

\usepackage{graphicx}
\usepackage{booktabs}

\usepackage[top=1in, bottom=1in, left=1.25in, right=1.25in]{geometry}
\usepackage{tcolorbox}
\usepackage{lipsum}

 \newcommand{\tnote}[1]{\textcolor{magenta}{\textbf{Tushant:} #1 }}  
 
\allowdisplaybreaks
\DeclareUnicodeCharacter{00A0}{ }
\begin{document}

\title{Derandomized Non-Abelian Homomorphism Testing in Low Soundness Regime}

 \author{
        Tushant Mittal\thanks{{\tt University of Chicago}. {\tt tushant@uchicago.edu}. Supported by NSF grant CCF-2326685.} \and
        Sourya Roy\thanks{{\tt The University of Iowa}. {\tt sourya-roy@uiowa.edu}.}    \and
}

\date{\today}

\newcommand{\projtriv}[1]{P^{{#1}\otimes{#1}^*}_{\mathrm{triv}}}
\date{}

\maketitle
\draftbox
We give a randomness-efficient homomorphism test in the low soundness regime for functions, $f: G\to \Utt$, from an arbitrary finite group $G$ to $t\times t$ unitary matrices. We show that if such a function passes a derandomized Blum--Luby--Rubinfeld (BLR) test (using small-bias sets), then (i) it correlates with a function arising from a genuine homomorphism, and (ii) it has a non-trivial Fourier mass on a low-dimensional irreducible representation. 
 
 { In the full randomness regime, such a test for matrix-valued functions on finite groups implicitly appears in the works of Gowers and Hatami~[{Sbornik: Mathematics} '17], and Moore and Russell~[{SIAM Journal on Discrete Mathematics} '15]. Thus, our work can be seen as a near-optimal derandomization of their results.} Our key technical contribution is a \enquote{degree-2 expander mixing lemma} that shows that Gowers' $\Ut$ norm can be efficiently estimated by restricting it to a small-bias subset. Another corollary is a \enquote{derandomized} version of a useful lemma due to Babai, Nikolov, and Pyber~[SODA'08] and Gowers~[Comb.\ Probab.\ Comput'08]. 	


 

\pagenumbering{arabic}
\setcounter{page}{1}

\section{Introduction}
An important problem in theoretical computer science is to efficiently test if a function $f: G \rightarrow H$  is correlated with some homomorphism between groups $G$ and $H$. Such tests are widely used in constructions of \textit{probabilistically checkable proofs }(PCPs), hardness of approximation, locally testable codes, and many other areas of computer science. Recently, there has been an interest in studying such tests for non-Abelian groups. For example, in quantum complexity, \textit{entanglement testing}~\cite{NV17} involves homomorphism testing over the (non-Abelian) Pauli group, which played an important role in the proof of MIP*=RE~\cite{MIP21}. Additionally, such non-Abelian tests have been used for constructions of better PCPs~\cite{BK21}, and hardness of approximation results~\cite{BKM22}.

The famous three-query randomized Blum--Luby--Rubinfield~\cite{BLR90} (BLR) test is as follows: pick two uniformly random group elements $x, y \in G$ and check if the homomorphism property holds for this pair, namely, if $f(x)f(y) = f(xy)$. This simple local test surprisingly sheds light on a global property of the function:  if a function $f:\Z_2^n\rightarrow \Z_2$ passes the test with non-trivial probability, then the function must have a non-trivial correlation with some homomorphism.  

This test can be used for any pair of groups $G, H$ (assuming that one can sample from $G$), and requires $2\log \abs{G}$ random bits. A randomness-efficient version of the test that has been studied is the \textit{derandomized BLR test} wherein $x$ is uniformly sampled from $G$ as before, but $y$ is chosen from a sparse pseudorandom set $S\subseteq G$. If $S$ is constant-sized, then the randomness is reduced to $\log |G| + O(1)$, which is almost optimal.

The study of such derandomized linearity (and low-degree) tests has found significant applications, particularly in the development of Probabilistically Checkable Proofs (PCPs). For example, the derandomization results in \cite{BSVW03} have enabled the construction of PCPs and Locally Testable Codes of nearly linear size. Additionally, derandomized parallel execution \cite{ST00, HW03} of the BLR test has facilitated the creation of PCPs with low amortized costs. We extend this study of derandomized tests and investigate the question,  
\begin{center}
\textit{Given a function $f: G\to H$ that passes the derandomized BLR test with probability $\delta$, what can one conclude about the function $f$?
}\end{center}

An ideal conclusion would be that the function $f$ is close to a true homomorphism $\varphi$ in some metric, \ie $\norm{f-\varphi} \leq \theta(\delta)$. This is achievable in the \enquote{$99\%$-regime}, when the test passing probability, $\delta$, is close to $1$. This is also called the \textit{unique-decoding regime}, as there is a unique homomorphism, $\varphi$, near the given function, $f$. The unique homomorphism can often be constructed via a \textit{majority decoding procedure}. There are many results in this setting~\cite{BCH+95, Farah00, badora2018approximate} including derandomized ones~\cite{SW04}. In particular, for any finite group $G$ and an arbitrary (not necessarily finite) group $H$, \cite{Farah00} constructs a homomorphism close in Hamming metric to the given function $f$, if the test passing probability is $\geq 10/11$.

 However, the situation is significantly more complex in the low-soundness or \enquote{$1\%$-regime}, wherein the function performs barely better than a random function, \ie the test passing probability is $\frac{1}{|H|} + \delta$. Firstly, one cannot always hope to find a homomorphism close in the Hamming metric. A folklore counterexample due to Coppersmith (also in \cite{BCLR07}) gives a function $f: \Z_{3^k}\to \Z_{3^{k-1}}$ that passes the test with probability $2/9$ but it is far away from every homomorphism in the Hamming metric. More interestingly, for some pair of groups $G,H$, the only homomorphism from $G\to H$ might be the trivial one. In this case, we might not be able to conclude that $f$ is close to the trivial homomorphism but still deduce something about the global structure of $f$. To do so, however, we need to have a better understanding of how the set of functions from $G$ to $H$, relates to the set of homomorphisms. For instance, Fourier analysis yields that any function $f: \Z_2^n\to \C$ can be expressed as a linear combination of homomorphisms. In general, \textit{representation theory} gives a similar relation for the more general setting of functions $f: G\to \Utt$, where $G$ is any finite group, and $\Utt$ is the group of $t\times t$ unitary matrices. This setting is, therefore, a natural starting point to start investigating the general question of derandomized testing for non-Abelian groups. Moreover, this setting (which we work with throughout our paper) has interesting connections to quantum linearity testing!

 \subsection{Our Setup: Matrix-valued functions}
We will work with functions from an arbitrary finite group $G$ to the group of $t\times t$ unitary matrices, $f: G\to \Utt$, and will use the following inner product to measure correlation, \ie  $\ip{f}{g}_\tr = \Ex{x \sim G}{\tr\parens{g^*(x)f(x)}}$. 
 We wish to design a randomness-efficient variant of BLR such that if a function $f: G \to \Utt$ passes such a test, then the function $f$ correlates with a homomorphism, or a functions arising from a homomorphism.

This setting has been studied in prior works~\cite{MR15, NV17, GH17}, most importantly in the context of quantum low-degree tests. {The results of~\cite{MR15} and \cite{GH17} are particularly relevant to our result, and we will discuss them in detail shortly. The other result by Natarajan and Vidick~\cite{NV17} gives a BLR-like test for homomorphism testing of functions, $f: \mathcal{P}^n\to \Utt$ where $\mathcal{P}^n$ is the $n$-fold tensor product of the \textit{Pauli group} (also known as the \textit{Weyl-Heisenberg} group). This was initially developed for entanglement testing~\cite{NV17}, and later became a crucial component in the $\text{MIP}^*=\text{RE}$ proof~\cite{MIP21}. 

While our setting encompasses their setup and has identical notions of correlation, our results do not directly apply to the quantum linearity test. This is because their test works with a specific presentation of the Pauli group due to additional constraints related to quantum measurements. Nevertheless, there might be interesting connections between our results and those in quantum homomorphism testing.

Before we state our results, we very briefly define some relevant concepts and discuss the challenges associated with this setting. See~\cref{sec:prelim} for detailed definitions.

 A \textit{group representation} is a pair $(\rho, V)$ where $\rho: G\to \U_V $ is a homomorphism from the group $G$ to the group of unitary operators on the complex Hilbert space $V$. If $V$is idendified with $\C^t$, we use the notation, $\U_t$. Every finite group $G$ has a finite set of \textit{irreducible representations} (irreps) $\widehat{G}$, which are the building blocks of complex-valued functions on $G$. In the case of Abelian groups, all the irreducible representations are one-dimensional and given by the Fourier characters. These characters also form an orthogonal basis for the space of complex-valued functions. For general finite groups, the orthogonal basis is given by the set of matrix coefficients of irreps, \ie $\braces[\big]{\rho(x)_{i,j} \mid \rho \in \widehat{G}\; i,j \leq \dim(\rho)}$. Just as in the Abelian case, we use $\hat{f}(\rho) = \Ex{x}{f(x) \rho(x)}$ to represent the coefficient corresponding to irrep $\rho$, which is now a matrix. Such a basis also exists for matrix-valued functions, $f: G\to \C^{t\times t}$, which is a collection of $t^2$ scalar functions.

\subsubsection{Challenges with this general setting}

 This general setting has three key differences with the original setting of BLR $f:\Z_2^n \to \Z_2$: (i) $G$ is an arbitrary (not necessarily Abelian) finite group, (ii) $H$ is continuous and not discrete, and (iii) the test passing probability is low (low-soundness regime). While each of these generalizations presents its own challenges, these are compounded when they are all together.   In order to clearly illustrate the issues, let us focus only on the case of complex-valued functions, $f: G\to \U_1 \subseteq \C$. The entire discussion is relevant when the functions are matrix-valued, but this special case captures all the difficulties.

\begin{description}
	\item[Hamming norm is unsuitable] Since the codomain is continuous, the Hamming norm is inappropriate as it is sensitive to small perturbations. For example, let $G$ be a finite group such that the only homomorphism to $\mathbb{U}_1$ is the trivial one. These groups exist and are known as \textit{quasirandom groups}. If $f(x)$ is set to $e^{-i\epsilon}$ for half of the inputs from $G$ and $e^{-2i\epsilon}$ for the rest, the BLR test passes with probability roughly $\frac{1}{8}$. But $f$ has a normalized Hamming distance of 1 from the closest homomorphism, \ie the trivial homomorphism. However, $f$ is actually close to the trivial homomorphism in $L^2$-distance, $\norm{f-g}^2 = \Ex{x}{\abs{f(x)-g(x)}^2}$. This suggests why previous works~\cite{GH17, BFL03, MR15} in this setting have used the $L^2$-norm.	
	\item[Need to look at larger representations] For Abelian groups $G$, every function $f:G\to \C$ can be expressed as a linear combination of homomorphisms from $G \to \C$. This is no longer true when $G$ is non-Abelian. As we have seen in the preliminaries, we need to rely on homomorphisms $\rho: G\to \U_t$ for $t$ potentially as large as $\sqrt{|G|}$, even though the original function maps to scalars. Thus, it is not immediate how to formalize the statement  \enquote{$f$ correlates with the homomorphism $\rho$,} as $f: G\to \C$ whereas, $\rho: G \to \U_t$ for $t >1$. There have been two non-equivalent solutions to this in prior work, 
	\begin{enumerate}
		\item Clip\footnote{For functions $f: G\to \U_t$ for $t >1$, the representations that need to be considered could be of dimension $t' < t$ as well. We still refer to it as \enquote{clipping}.	
} a representation -- Let $g_\rho : G\to \C$ be defined as $g_\rho(x) = V\rho (x) U^*$ for a homomorphism $\rho$, and $V, U \in \C^{1\times t}$. One can now check if $f$ correlates with $g_\rho$. This is the route taken by Gowers and Hatami~\cite{GH17}.  
		\item Large Fourier Mass -- If the representation $\rho$ is irreducible, one can look at how much of the function can be explained by the Fourier basis elements corresponding to $\rho$, \ie $\{\rho_{i,j}\}_{i,j}$. Note that these elements are no longer homomorphisms (unlike for Abelian groups), but $\norm{\hat{f}(\rho)}_\HS$ being large is another way to formalize $f$ being correlated with $\rho$, as $\norm{f}^2 = \sum_\rho \dim(\rho)\norm{\hat{f}(\rho)}_\HS^2$ by Parseval's. This approach is followed by \cite{MR15}.   
	\end{enumerate}

	\item[Representation sizes depend on test-passing probability] Ideally, we would want to show the correlation of the function with a representation $\rho$ of dimension $1$. But the above discussion shows why that is too much to ask for. Nevertheless, one might still want to bound the dimension of these representations, the intuition being that the smaller the dimension, the closer $f$ is to being a true homomorphism. We will see that this can be done but this dimension must depend on the test passing probability $\delta$. Such a dependence is unavoidable, as the following counterexample illustrates. Let $\Gamma$ be a non-Abelian group containing an irrep of dimension $d \simeq \poly(|\Gamma|)$, say $\rho$. Consider the group $G = \Z_2^n \times \Gamma$ and its irrep, $\psi = \triv\, \otimes \rho $. The function $f(x) = \psi(x)_{1,1}$ passes the randomized BLR test if either of the query points is in the kernel of this irrep. Since $\Z_2^n$ lies in the kernel, the test passes with probability at least $\frac{1}{\abs{\Gamma}}$. However, by construction, $f$ is entirely supported on a $d$-dimensional irrep, $\psi$.
\end{description}

\subsubsection{Known Results}

As previously mentioned, despite the outlined challenges, researchers have achieved intriguing results in this matrix-valued function setting that we are interested in. Specifically, Moore and Russell~\cite{MR15} as well as Gowers and Hatmai~\cite{GH17} explored the homomorphism testing problem for functions from any finite group $G$ to $\Utt$. It is important to note that their studies did not explicitly focus on this question from a testing perspective. Nevertheless, their findings can be easily adapted into a BLR-type linearity testing framework. Additionally, when translated into homomorphism testing terminology, their approaches correspond to the same Hilbert-Schmidt version of the BLR test that we employ. Below, we summarize the homomorphism testing results derived from these two studies.

\begin{theorem}[Tests from \cite{GH17, MR15}]
	Let $G$ be any finite group, and $f: G\rightarrow \Utt$ be a unitary matrix-valued function. Assume that the function $f$ passes the BLR test with probability $ \delta$. Then,
	\begin{enumerate}
		\item ({Correlation with clipped representation~\textup{\cite{GH17})}} :  There is a representation, $\pi: G\rightarrow \U_{t'}$ and two matrices, $V,U \in \C^{t\times t'}$, such that $f$ correlates with the function, $g_\pi = V\pi(x)U^*: G\to \Utt $ :
		\[{\angles[\big]{f, g_\pi }_\tr}  ~\geq~ \tfrac{\delta^2}{4}\cdot t
		\] 
		Moreover, the representation has bounded dimension, $ t' = \dim(\pi) \in {\brackets[\big]{\delta^2 t, \frac{2t}{\delta^2}}}$. 
		\item (Fourier Mass on a low-dimensional irreducible representation~\textup{\cite{MR15})} : There exists an irrep $\rho \in \widehat{G}$ such that $\dim(\rho) < \frac{2t}{\delta^2}$ and $\|\widehat{f}(\rho)\|^2_{\HS} \geq \frac{\delta^2}{2}$. 
	\end{enumerate}
\end{theorem}

\begin{remark}
	Note that the representation, $\pi$, in part 1 of the theorem is not guaranteed to be irreducible, but the second one is.  
\end{remark}

 \subsection{Our Results}
 
 The main contribution of this work is to give a derandomized BLR-like homomorphism test in the low soundness regime for the general setup of functions from an arbitrary finite group $G \to \Utt$. Prior to this work, the only known derandomized test in the $1\%$-regime is that of~\cite{BSVW03} for the case when $G = \Z_p^n$ and $H = \Z_p$.

  Our key derandomization tool is  \textit{small--bias sets}, which are those that \enquote{fool} irreducible representations. Formally, a set $S \subseteq G$ is
  $\epsilon$-biased if for every non-trivial irreducible
  representation $\rho$, we have $\opnorm{ \Ex{s\sim S}{\rho(s)} }
  \leq \epsilon$. Here, the operator norm is the largest singular value of the operator. In the Abelian case, the irreducible representations are \emph{characters}, and thus, this definition is a generalization of the usual one of \emph{fooling} non-trivial characters~\cite{NN93,AGHP92}.
 We have the following derandomization of a result of Alon--Roichman~\cite{AR94}, 
 
 \begin{theorem}[{\cite[Thm 5.1]{WX08}}]\label{thm:alon}
 	For every finite group $G$ and any constant $\ep > 0$, there exists a deterministic $\poly(|G|)$-time algorithm that outputs an $\ep$-biased set $S\subseteq G$ of size $|S|\leq O\parens[\Big]{\frac{\log |G|}{\ep^2}}$. 
 \end{theorem}

 While this bound is tight over Abelian groups, we can do much better for other groups. Particularly for all \textit{finite simple groups} we now have explicit constant-sized small-biased sets due to a long line of work~\cite{KN06, Kas07, Lub11}. These can also be made near-optimal using the amplification machinery in~\cite{JMRW22}. 
 
 \begin{theorem}[{\cite{KN06, Kas07, Lub11, JMRW22}}]\label{thm:simple}
 	For every non-abelian finite simple group $G$ and any  $\ep > 0$, there exists a deterministic $\poly(1/\ep)$-time algorithm that outputs an $\ep$-biased set $S\subseteq G$ of size $|S|\leq O\parens[\big]{{\ep^{-(2+o(1))}}}$. 
 \end{theorem}

 \paragraph{Result 1: Derandomized Homomorphism Testing} 
 We analyze the following derandomized variant of BLR. Here, the parameter $\gamma$ allows for a relaxed version of this test that makes the test robust to small noise, which can be useful as $\Utt$ is a continuous group.  
 \begin{center}
 	 \fbox{\begin{minipage}{12 cm}
		~~
		\\
		{\bf Derandomized BLR}$_{\gamma}(G,S,f)$:
		\begin{enumerate}
			\item Sample $x\sim G, y\sim S$. 
			\item If $\norm{f(xy)-f(x)\cdot f(y)}^2_{\HS} ~\leq~ \gamma t$, output {\it Pass}. Else, output {\it Fail}. 
			~~
			\\
		\end{enumerate}	
\end{minipage}}
 \end{center}

Setting $\gamma = 0$ recovers the usual derandomized version of the BLR test that has been used in previous derandomizations of homomorphism tests~\cite{BSVW03, SW04}. The following table compares our derandomized test with other tests.

\begin{table}[h]
	\begin{center}
		\small
		\begin{tabular}{cccc}
			\toprule
			Work & Setting ($f: G\to H$)  & Conclusion  & Randomness
			\\[2pt]
			\midrule
			\multicolumn{4}{c}{High Soundness}\\[1.9pt]
			\midrule
			\cite{BLR90} & $G = \Z_2^n,\, H = \Z_2$ & Hamming & $2\log |G|$	\\
			\cite{BCLR07} & $G,H$ any finite groups & Hamming & $2\log |G|$	\\
			\cite{SW04} & $G,H$ any finite groups & Hamming & $(1+o(1))\log |G|$	\\[2pt]  \midrule
			\multicolumn{4}{c}{Low Soundness}\\[1.8pt]
			\midrule
			\cite{BCH+95} & $G = \Z_2^n,\, H = \Z_2$ & Hamming & $2\log |G|$	\\
			\cite{Kiwi03} & $G = \Z_p^n,\, H = \Z_p$ & Hamming & $(2+o(1))\log |G|$	\\
			\cite{BSVW03} & $G = \Z_p^n,\, H = \Z_p$ & Hamming & $(1+o(1))\log |G|$	\\
			\cite{Sam07, Sanders10, GGMT23} & $\Z_2^n \to \Z_2^m$ & Hamming & $2\log |G|$	\\
			\cite{BFL03} & $G$ finite Abelian,  $H = \U_1$ & Correlation & $2\log |G|$	\\
			\cite{MR15} & $G$ any finite group,  $H = \Utt$ & Correlation & $2\log |G|$	\\
			\cite{GH17} & $G$ any finite group,  $H = \Utt$ & Hilbert-Schmidt  & $2\log |G|$	\\
			Our Result & $G$ any finite group,  $H = \Utt$ & Hilbert-Schmidt & $(1+o(1))\log |G|$\\
			Our Result & $G$ any finite group,  $H = \Utt$ & Correlation & $(1+o(1))\log |G|$\\
			\hline
		\end{tabular}\caption{A summary of prior works on homomorphism testing}
	\end{center}
\end{table}

\begin{theorem}[Informal version of~\cref{thm:main}]\label{thm:intro_main}
	Let $G$ be any finite group, and $f: G\rightarrow \Utt$ be a unitary matrix-valued function. Let $S\subseteq G$ be an $\ep$-biased set. Assume that the function $f$ passes the derandomized BLR test with probability $ \delta > \sqrt{\ep}$.
Then,
	\begin{enumerate}
		\item ({Correlation with clipped representation}): There is a representation, $\pi: G\rightarrow \U_{t'}$ and two matrices, $V,U \in \C^{t\times t'}$, such that for $g_\pi = V\pi(x)U^*: G\to \Utt $, $f$ correlates with $g_\pi$,
		\[{\angles[\big]{f, g_\pi }_\tr}  ~\geq~ \tfrac{\delta^2 -\ep}{4}\cdot t
		\] 
		Moreover, the representation has bounded dimension, $ (\delta^2 -\ep) t \leq t' = \dim(\pi) \leq \frac{2t}{\delta^2-\ep}$. 
		\item ({Fourier Mass on a low-dimensional irreducible representation}):  There exists an irrep $\rho \in \widehat{G}$ such that $\dim(\rho) < \frac{2t}{\delta^2-\ep} $ and $\|\widehat{f}(\rho)\|^2_{\HS} \geq \frac{\delta^2-\ep}{2}$. 
	\end{enumerate}
	Moreover, if one uses the $\gamma$-robust BLR test, the same conclusions hold with $\delta$ replaced by $\delta-(\gamma/2)$. 
\end{theorem}

Using the small-bias set construction from~\cref{thm:alon}, we get a test that uses $\log |G| + \log |S| =  \log |G| + O\parens{\log \log |G|} =  (1+o(1))\log |G|$-random bits. For special families of groups like the class of \textit{finite simple groups}, we can use~\cref{thm:simple} to further reduce the randomness to $\log |G| + O(1)$, which is almost optimal.

\paragraph{Result 2: Derandomized BNP Lemma}  The \enquote{BNP lemma} is a very useful observation due to Babai, Nikolov, Pyber~\cite{BNP08}, and Gowers~\cite{Gow08}. This lemma gives an improvement over \CS \, for \textit{quasirandom groups}, \ie groups with no small non-trivial irreps, and has been used to analyze mixing in progressions~\cite{BHR21}, product-free sets~\cite{Gow08}, and hardness of approximation~\cite{BKM22}, to name a few. In its most general form, it says that for functions $f, g: G\to \Mt(\C)$ that map to $t\times t$-matrices, we have,
\[  \norm{f*g}^2 =   \Ex{s\sim G}{\hsnorm{(f*g)(s)}^2} ~\leq~ \frac{1}{D}\,\norm{f}_2^2 \norm{g}_2^2 \]

We show that such a bound holds even when the average is over a small-bias set $S\subseteq G$, which could be of constant size for some groups! 

\begin{restatable}[Derandomized Matrix BNP]{lemma}{bnp}\label{lem:bnp} 
Let $G$ be a group such that the dimension of the smallest non-trivial irrep is $D$ and let $S\subseteq G$ be an $\ep$-biased set. Let $f,g: G\to \Mt(\C)$ be mean-zero functions. Then,	
\[  \Ex{s\sim S}{\hsnorm{(f*g)(s)}^2} ~\leq~ \parens[\bigg]{\frac{1}{D}+\ep}\,\norm{f}_2^2 \norm{g}_2^2 \]
 The usual BNP lemma can be recovered by setting $S = G$, and thus, $\ep = 0$. 
\end{restatable}

 \subsection{Technical Overview}

Our main conceptual contribution is to initiate the study of the non-Abelian generalization of two key notions in the analysis of Boolean functions: (i) \textit{spectral norm} of a function and (ii) \textit{spectral positivity}.

\paragraph{Spectral norm and its non-Abelian analog}

The $\ell_1$-norm of the Fourier transform of a function is known as its \textit{spectral norm}. Spectral norm has emerged as an important quantity for the analysis of Boolean functions, \ie functions over $\Z_2^n$. In particular, functions with low spectral norm have a lot of structure~\cite{STV17}: they admit small decision trees, parity decision trees, they are easily learnable, etc.
 
One of the conceptual contributions of this paper is studying the non-Abelian analog of this norm from the perspective of pseudorandomness. A first generalization one can think of would be a similar $\ell_1$ norm of the Fourier coefficients. However, it turns out that the appropriate generalization of the spectral norm is the \textit{Fourier algebra norm}. This was suggested earlier by Sanders~\cite{Sanders21}, who used it to generalize the quantitative idempotent theorem. This norm has multiple equivalent definitions, but our key idea is to use the following harmonic analytic reformulation due to Sanders~\cite{Sanders21} (attributed to~\cite{Eymard64}),
\[ \norm{f}_A = \min_{(\pi,V)} \braces*{\,\norm{u}\cdot\norm{v} \,\big\lvert\; f(x) = \ip{u}{\pi(x)\, v} \,} \] where $(\pi,V)$ is a representation of $G$  and $u, v\in V$. \footnote{The Fourier inversion theorem gives one such an expression for $f$ by using the \textit{regular representation}; although it might not be the one that minimizes the algebra norm, and hence one minimizes over such expressions.}


It is well-known that any function, $f$, on an Abelian group is $\epsilon\norm{\hat{f}}_1$-fooled by any $\ep$-biased set. We show that this neatly generalizes to any finite group by replacing the spectral norm with Fourier algebra norm, any function, $f$, on a finite group is $ \epsilon\norm{f}_A$-fooled by any $\ep$-biased set. 

\paragraph{Spectral Positivity and its non-Abelian analog} A function over an Abelian group $G$, $f: G\to \C$, is spectrally non-negative of $\hat{f}(\chi) \geq 0$ for every character $\chi$. This notion played a key role in the recent breakthrough by Kelley and Meka~\cite{KM23} on 3-AP free sets.

This naturally generalizes to the finite group setting wherein a \textit{positive-definite functions} is a function $f$ such that $\hat{f}(\rho)$ is positive semi-definite for every irreducible representation $\rho$. The important observation is that such functions have small algebra norm, $\norm{f}_A = f(1)$. We use this observation to prove that small--bias sets can be used to approximate the $\Ut$-norm.

\subsubsection{Proof Overview}

Denote $\tilde{f}(x) = f(x^{-1})^*$, and recall the following two norms,
\begin{align*}
(\Ut\text{-norm})& \;\;\;
\norm{f}_{\Ut}^4 = \norm{f*\tilde{f}}^2 = \Ex{x\sim G}{ \norm{(f*\tilde{f})(x)}_\HS^2} \\
(\text{Algebra norm})& \;\;\; \norm{f}_A = \min_{(\pi,V)} \braces*{\,\norm{u}\cdot\norm{v} \,\big\lvert\; f(x) = \ip{u}{\pi(x)\, v}}  
\end{align*}

We now give a quick summary of the key steps involved in the proof:
\begin{enumerate}
	\item  (\cref{lem:pd_fooled}) Small--bias sets fool functions with a {small  algebra norm}.
	\item (\cref{lem:matrix_conv}) Let $f, g : G\to \Utt$ be any functions. Then, the function, $x \mapsto \norm{(f*g)(x)}_{\HS}^2$ has a small Fourier algebra norm.
	\item The above two lemmas imply a degree-2 EML. This immediately yields our result on the derandomized BNP lemma~(\cref{lem:bnp}). We expect that this degree-2 EML will have uses beyond this work, and we explain this below.
	\item A special case of the above EML implies that small bias sets approximate $\Ut$-norm (\cref{cor:u2}). Thus, the test passing probability of the derandomized test implies a large $\Ut$-norm of the function. Combining this with the inverse theorem of Gowers-Hatami~\cite{GH17} gives us the first part of~\cref{thm:intro_main}.
	\item The second part of~\cref{thm:intro_main} follows from the same large $\Ut$-norm consequence implied by test passing. To achieve this, we adapt the proof strategy of the BNP lemma~\cite{BNP08} to our setup, which relies on basic non-Abelian harmonic analysis.
\end{enumerate}


\paragraph{Degree-2 EML} Our key technical contribution is a degree-2 variant of the celebrated \textit{expander mixing lemma} (EML).  Recall that EML characterizes spectral expansion. When applied to the Cayley graph $\Cay(G,S)$, we get that $S$ is an $\ep$-biased set if and only if the EML holds,
 \[
	\abs[\Big]{\Ex{s\sim S}{{(f*g)(s)}} - \Ex{s\sim G}{{(f*g)(s)}}} ~\leq~ \ep \norm{f}_2 \, \norm{g}_2, \;\;\;\; 	\text{(EML)} \mper
\]

 We prove that such sets also satisfy a degree-2 variant of the above inequality,   
 \[
	\abs[\Big]{\Ex{s\sim S}{\norm{(f*g)(s)}^2} - \Ex{s\sim G}{\norm{(f*g)(s)}^2}} ~\leq~ \ep \norm{f}_2^2 \, \norm{g}_2^2, \;\;\;\; 	\text{(Our degree-$2$ EML)} 	\mper
\]

Using it for the special case of where $g(x) = \tilde{f}(x) = f(x^{-1})^*$, we get that small bias sets approximate $\Ut$-norm.

\begin{corollary}[Small bias sets approximate $U^2$-norm]\label{cor:u2} For $f: G\to \Utt$, \[
	\abs[\Big]{\Ex{s\sim S}{\norm{(f*\tilde{f})(s)}^2} - \norm{f}_{\Ut}^4 } ~\leq~ \ep\cdot t  	\mper
\]
Thus, the $\Ut$-norm of a function $f$ can be $\ep$-estimated by querying $f*\tilde{f}$ on an $\ep$-biased set $S$. 	
\end{corollary}

\subsection{Related Work}
\paragraph{High soundness regime} Blum--Luby--Rubinfield~\cite{BLR90} analyzed linearity tests for functions of the form $f:\Z_2^n \to \set{\pm{1}}$. This was extended to the setting $f: G \to H$ where both are arbitrary finite groups, by Ben-Or, Coppersmith, Luby, and Rubinfeld~\cite{BCLR07}. This result was derandomized by Shpilka and Wigderson~\cite{SW04}. Going beyond finite groups, Farah~\cite{Farah00}, and later, Badora and Przebieracz~\cite{badora2018approximate}, give homomorphism tests for any \textit{amenable group} $G$, and any group $H$,  equipped with an invariant metric. 
 
\paragraph{Low soundness regime}
Bellare, Coppersmith, H{\aa}stad, Kiwi, and Sudan~\cite{BCH+95} analyzed linearity tests for functions of the form $f:\Z_2^n \to \set{\pm{1}}$ in this low-soundness regime. This was extended to the setting of $\Z_p^n \to {\Z_p}$ by H\aa stad and Wigderson~\cite{HW03}. This result was derandomized using $\ep$-biased sets by Ben-Sasson, Sudan, Vadhan, and  Wigderson~\cite{BSVW03}. For the same setting, Kiwi~\cite{Kiwi03} analyzed a variant of the BLR test that uses a lot more randomness but gives an improved correlation. Samorodnitsky~\cite{Sam07} studied a completely different setup where $H$ is large and not a subset of $\C$. He showed that if a function $f:\Z_2^n \to \Z_2^m$, passes the test with probability $\delta$, then it has an exponentially small agreement with a homomorphism.  Improving the agreement to polynomial in $\delta$ is equivalent to the polynomial Freiman--Rusza (PFR) conjecture which was finally settled recently~\cite{Sanders10, GGMT23}. 

 \section{Preliminaries}
\label{sec:prelim}

Throughout the work, we will assume that $G$ is a finite group, and all vector spaces are finite-dimensional over $\C$. We work with vector spaces of matrix-valued functions, and unless mentioned otherwise, these spaces are equipped with the expectation inner product, $\ip{f}{g} = \Ex{x}{\ip{f(x)}{g(x)}_\tr}$.  We use $\Utt$ to denote the group of $t\times t$ unitary matrices, and $\Mt(\C)$ to denote the set of $t\times t$ complex matrices.

\subsection{Group Representations and small-bias sets}

For finite groups, every representation can be made unitary, and thus, we can restrict to
studying these. Let $V$ be a complex Hilbert space and denote by
 $\mathbb{U}_V$, the {group of unitary operators} acting on $V$. 
 
\begin{definition}[Unitary Group Representation]
  For a group $G$, a unitary representation is a pair $(\rho, V)$
  where $\rho:G\to \mathbb{U}_V$ is a group homomorphism, \ie for
  every $x,y \in G$, we have $\rho(xy) = \rho(x)\rho(y)$.
  A representation is \textit{irreducible} if the only subspaces of $V$ that are invariant
  under the action of $\rho(G)$ are the empty space, $\set{0}$, and the entire space, $V$. For a representation $(\rho, V)$, will use $d_\rho$ to denote $\dim(V)$. 
\end{definition}

We use $\widehat{G}$ to denote the set of all irreducible representations (\textit{irreps}) of a group $G$. Every group has a special irreducible representation called the \textit{trivial representation},  $(\rho_\triv, \C)$, where $\rho(g) = 1$ for every group element $g \in G$. The following is a fundamental result that states that every representation decomposes as a finite sum of irreducible ones.

\begin{theorem}[Maschke]\label{thm:peter}
	Let $G$ be a finite group and let $(\pi, V)$ be any representation of $G$. Then,  $V = \oplus_i V_{\rho_i}$, \ie it decomposes as a direct sum of irreducible representations $\braces{\rho_i}_i$. Explicitly, there exists a unitary transformation $U_\pi$ such that $U_\pi \pi U_\pi^*$ is block-diagonal with each block being $\rho_i$.
\end{theorem}

\paragraph{Small--bias sets} Small bias sets over Abelian groups were introduced in the pioneering work of Naor and Naor~\cite{NN93}, and are a fundamental derandomization tool that has been widely used across domains like complexity theory, coding theory, learning theory, graph theory, etc. See the excellent surveys~\cite{HH24, HLW06} for references and details. Apart from the wealth of applications arising from expanders, $\ep$-biased sets for non-Abelian groups have also recently found applications in constructions of near-optimal expanders and quantum expanders~\cite{JMRW22}, constructions of \textit{unitary designs}~\cite{OSP23}.

\begin{definition}[$\delta$-fooled]
A function $f: G\to \Mt(\C)$ is $\delta$-fooled by a set $S\subseteq G$ if,
\[\norm[\Big]{\;\Ex{x\sim G}{f(x)} - \Ex{x\sim S}{f(x)}\;}_{\mathrm{op}} \leq \delta .\]	
\end{definition}

Small--bias sets are those that fool all non-trivial irreducible representations.  In the Abelian case, the irreducible representations are \emph{characters}, and thus, this definition is a generalization of the usual one of \emph{fooling} non-trivial characters~\cite{NN90,AGHP92}. 

\begin{definition}[$\eps$-Biased Set]\label{def:bias_group}
  Let $\epsilon \in [0,1)$.
  We say that a multiset $S \subseteq G$ is
  $\epsilon$-biased if for every irreducible
  representation $\rho$, $\rho$ is $\ep$-fooled by S, \ie $\opnorm{ \Ex{s\sim S}{\rho(s)} } \leq \epsilon$. Here, the operator norm is the largest singular value of the operator
 \end{definition}

The notion of small bias is closely linked to graph expansion, in particular, a symmetric subset $S\subseteq G$ is an $\eps$-biased set if and only if the $\Cay(G,S)$ is an $\eps$-spectral expander, \ie the adjacency matrix has second largest eigenvalue $\epsilon$.

\subsection{Matrix--valued functions and $\Ut$-norm}
Denote by $L^2_t(G) = \{ f: G\to \Mt(\C) \}$, the space of $t\times t$ matrix-valued functions equipped with the trace expectation inner product, 
\begin{equation}\label{eqn:ip_trace}
\ip{f}{g} = \Ex{x\sim G}{\ip{f(x)}{g(x)}_{\tr}} = \Ex{x\sim G}{\Tr\parens[\big]{g(x)^*f(x)}}
\end{equation}
The induced norm is $\norm{f}^2 =  \Ex{x\sim G}{\hsnorm{f(x)}^2}$. 
 For a function $f$, we denote its \textit{adjoint} by $\tilde{f}(x) := f(x^{-1})^*$. The operation of convolution generalizes as, 
\[(f*g)(x) :=\E_{y\sim G}\brackets[\big]{f(xy^{-1})g(y)} = \E_{y\sim G}\brackets[\big]{\tilde{f}(y)^* g(yx)}.
\]

\begin{definition}[Matrix Fourier Coefficient]
	For any irrep $\rho$, we have $\widehat{f}(\rho):=\E_{x}\brackets[\big]{f(x)\otimes \rho(x)}.$ We denote the coefficient of the trivial irrep as  $\mu(f) := \widehat{f}(\rho_\triv)$.
\end{definition}

\begin{fact}
	\label{fact:fourier}
The following identities hold for the matrix Fourier transform,
\begin{enumerate}
\item {\bf (Parseval's identity)} $\;\;\norm{f}^2 = \E_{x}\brackets[\big]{\hsnorm{f(x)}^2}=\sum_{\rho \in \widehat{G}}d_{\rho}{\norm{\widehat{f}(\rho)}^2_{\HS}}$
\item {\bf (Convolution identity)} $\;\widehat{f*g}=\widehat{f}\cdot {\widehat{g}}$
\item {\bf ($\Ut$ norm)} $\;\;\norm{f}_{\Ut}^4 = \norm{\tilde{f} * f}^2 =\sum_{\rho} d_{\rho}\norm[\big]{\widehat{f}(\rho)\widehat{f}(\rho)^*}^2_{\HS}=
	\sum_{\rho} d_{\rho}\norm[\big]{\widehat{f}(\rho)^*\widehat{f}(\rho)}^2_{\HS}$
\end{enumerate}
\end{fact}
\begin{proof}
These facts are simple extensions of the scalar-valued functions and were considered in~\cite{BFL03, MR15, GH17}. Since the last one is perhaps atypical, we provide a quick proof.
\begin{align*}
	\norm{f}_{\Ut}^4 ~&:=~ \Ex{xy^{-1}= wz^{-1} }{\Tr\parens[\big]{\,f(x)f(y)^*f(z)f(w)^*\,}}\\
	~&=~ \Ex{t = xy^{-1}= wz^{-1} }{\ip{f(x)f(y)^*}{f(w)f(z)^*\,}}\\
	~&=~ \Ex{t = xy^{-1}= wz^{-1} }{\ip{f(x)\tilde{f}(y^{-1})}{f(w)\tilde{f}(z^{-1})\,}}\\
	~&=~\;\;\;\;\;\;\;\, \Ex{t}{\ip{\,(f*\tilde{f})(t)}{(f*\tilde{f})(t)\,}} = \norm{f*\tilde{f}}^2
\end{align*}
The second equality follows from Parseval's and convolution identity once one observes that, 
\[\widehat{\tilde{f}}(\rho) = \Ex{x}{f(x^{-1})^*\otimes \rho(x)} =  \Ex{x}{f(x)^*\otimes \rho(x^{-1})} = \Ex{x}{f(x)^*\otimes \rho(x)^*} = \hat{f}(\rho)^* \qedhere
 \]
\end{proof}

\subsection{Fourier algebra norm and positive definite functions}
Let $f: G\to \C$ be any function and let $T_f$ be the convolution by f operator, $T_f(g) = f*g$. More explicitly, $T_f(x,y)= \frac{1}{\abs{G}}f(x^{-1}y)$ is a $\abs{G}\times \abs{G}$ matrix.

\begin{definition}[Fourier algebra norm]\label{def:alg_norm}
	The \textit{algebra norm} of $f: G\to \C$ has the following equivalent definitions,
	\begin{itemize}
		\item[i.]  $\norm{f}_A = \sup \braces[\big]{\ip{f}{g} \;\big\lvert\; \norm{T_g}_{\mathrm{op}} \leq 1 } $.
		
			\item[ii.]  $\norm{f}_A = \norm{T_f}_{\mathrm{tr}} =  \sum_i \sigma_i(T_f)$ where $\set{\sigma_i}$ are the singular values of $T_f$.
	
		\item[iii.] $\norm{f}_A = \min_{(\pi,V)} \braces*{\,\norm{u}\cdot\norm{v} \,\big\lvert\; f(x) = \ip{u}{\pi(x)\, v} \,} $ where $(\pi,V)$ is a representation of $G$  and $u, v\in V$. 
	\end{itemize}
\end{definition}

The equivalence of definitions (i) and (ii) can be found in \cite[Lem. 5.2]{San11} and also in \cite[Prop 3.11]{HHH22}. The proof of equivalence between (i) and (iii) is present in \cite[Lem. 2.2]{Sanders21}, attributed to Eymard~\cite[Th\`eor\'em, Pg 218]{Eymard64}. In particular, Sanders shows that the minimum is indeed attained for some representation $(\pi,V)$.

\paragraph{Abelian Case} When the group is Abelian, the algebra norm coincides with the \textit{spectral norm}, \ie $\norm{f}_A = \norm{\,\hat{f}\,}_{1}$. This can be seen by observing that $T_f$ is a diagonal matrix (in the Fourier basis) with the Fourier coefficients on the diagonal. Therefore, the algebra norm is a generalization of the spectral norm to the non-Abelian setting.

\begin{definition}[Positive definite functions]
Let $G$ be a finite group. A function $f: G\to \C$, is said to be positive definite if the convolution operator, $T_f$, is positive semi-definite. 
\end{definition}

The following simple observation states that positive definite functions have a small algebra norm. We will crucially use this later to bound the algebra norm of a function.
\begin{observation}\label{obs:pd_algnorm}
	If a function $f$ is positive-definite, then $\norm{f}_A = f(1)$.
\end{observation}
\begin{proof}
	Since $T_f$ is positive semi-definite, $\norm{f}_A = \norm{T_f}_{\tr} = \tr(T_f) = f(1)$ as $T_f(x,y)= \frac{f(x^{-1}y)}{\abs{G}}$. 
\end{proof}

\section{Small-bias sets \enquote{fool} small norm functions}

In the Boolean setting, properties of low spectral norm function have been well-studied \cite{KM93, STV17}. Green and Sanders~\cite{GS08} showed that functions with low spectral norm can be expressed as a $\pm 1$ combination of characteristic functions of cosets. Sanders \cite{San11} generalized it to non-Abelian groups with spectral norm replaced by the algebra norm. 

This suggests that the algebra norm is indeed the right generalization of the spectral norm, and one might investigate other properties of functions with low algebra/spectral norm. In this section, we make yet another connection by showing that small-bias sets fool functions with small algebra norm; again generalizing the Abelian case.

\paragraph{Boolean Cube} Let $f:\Z_2^n \to \C$, be any function and let $S$ be an $\epsilon$-biased set. Then,
\begin{align*}
\abs[\Big]{\Ex{h\sim S}{f(h)} - \Ex{h\sim G}{f(h)}}  ~&=~ \abs[\Big]{\sum_{\chi} \hat{f}(\chi)\parens[\Big]{\Ex{h\sim S}{\chi(h)} - \Ex{h\sim G}{\chi(h)}} }\\
~&\leq~  \max_\chi\abs[\Big]{\Ex{h\sim S}{\chi(h)} - \Ex{h\sim G}{\chi(h)}} \cdot \sum_{\chi\neq 0} \abs[\big]{\hat{f}(\chi)}\\
 ~&\leq~ \epsilon \cdot \norm{f-\mu(f)}_A
\end{align*}

Moreover, this is tight due to a result of~\cite[Prop. 2.7]{DETT10}, which says that any function that is fooled by all $\ep$-biased sets must be sandwiched by low spectral norm functions.

\paragraph{General finite groups}

We now generalize the above result to finite groups. The key here is to use the harmonic analytic definition of the spectral norm which makes the proof surprisingly simple.

\begin{lemma}\label{lem:pd_fooled}
	Let $f:G \to \C$, be a function and $S\subseteq G$ be any $\epsilon$-biased set,. Then, 
\[
\abs[\Big]{\;\Ex{x\sim S}{f(x)} - \Ex{x\sim G}{f(x)}\;} \leq \epsilon \cdot \norm{f}_A\,\mper	
\]
\end{lemma}
\begin{proof} From~\cref{def:alg_norm}, we get that there is a representation, $(\pi, V)$, representation of $G$ such that $f(x) = \ip{u}{\pi(x)\,v}$ for some $u, v \in V$. Moreover, $\norm{u}\norm{v} \leq \norm{f}_A$.
Using Maschke's theorem~(\cref{thm:peter}), we have $\pi = \rho_{\triv}^{\oplus c}\oplus_i\rho_i$ where $c$ denotes the multiplicity of the trivial representation and $\rho_i$ are all non-trivial irreducible representations (possibly with repetitions). Let $U_\pi$ be the unitary transformation that block-diagonalizes $\pi$. Then, $U_\pi v = v_{\triv} \oplus v_i$ and similarly $U_\pi u = u_{\triv} \oplus u_i$. Thus, we have, 
\begin{align*}
 \ip{u}{\pi(x) \,v} = \ip{U_\pi u}{U_\pi\, \pi(x)\, v} &= \ip{U_\pi u}{\parens*{U_\pi \pi(x)U_\pi^*}\, U_\pi v}   \\
 ~&=~  \ip{u_{\triv} \oplus_i u_i}{v_{\triv} \oplus_i (\rho_i v_i)}\\
  ~&=~  \ip{u_{\triv}}{v_\triv} + \sum_i \ip{u_i}{\rho_i(x)\, v_i}
\end{align*}
  Now, $\ip{u_{\triv}}{v_\triv}$ is a constant and is the same for both terms. Therefore,
  
  \begin{align*}
\Ex{x\sim S}{f(x)} - \Ex{x\sim G}{f(x)} ~&=~ \sum_i \ip{u_i}{\Ex{x\sim S}{\rho_i(x)} - \Ex{x\sim G}{\rho_i(x)}\, v_i}  \\
\abs[\Big]{\Ex{x\sim S}{f(x)} - \Ex{x\sim G}{f(x)}}  ~&\leq~ \sum_i \norm{u_i}_2\cdot \norm[\Big]{\Ex{x\sim S}{\rho_i(x)}- \Ex{x\sim G}{\rho_i(x)}}_{\mathrm{op}} \cdot\norm{v_i}_2 &&(\text{From Cauchy--Schwarz}) \\
 ~&\leq~ \epsilon \sum_i \norm{u_i}_2\cdot\norm{v_i}_2 &&(S\text{ is an $\eps$-bias set}) \\
  ~&\leq~ \epsilon \cdot \sqrt{\sum_i \norm{u_i}_2^2} \cdot\sqrt{\sum_i \norm{v_i}_2^2} &&(\text{ Cauchy--Schwarz})\\
  ~&=~ \epsilon\cdot \norm{u - u_\triv}_2\cdot\norm{v-v_{\triv}}_2 &&(U_\pi \text{ is unitary})\qedhere
\end{align*}
\end{proof}

\begin{corollary}[PD functions are fooled]
	\label{cor: pdFooled}
	If $f:G\to \C$ is a positive-definite function, then $\norm{f}_A = f(1) \leq  \norm{f}_\infty$. Therefore, if $\norm{f}_\infty\leq 1$, then $f$ is $\epsilon$-fooled by every $\epsilon$-biased set.
\end{corollary}
\begin{proof}
Since, $f$ is positive definite, $T_f$ is PSD and thus, $\norm{f}_A = \norm{T_f}_\tr = \Tr(T_f) = f(1)\leq 1$.
\end{proof}

\subsection{$\Ut$ norm and algebra norm} 
\label{sec:normPD}
Let $f: G \to \matr M_t(\C)$ be a function, then $\norm{f}_{\Ut}^4 = \Ex{y\sim G}{\psi(y)}$ where,  
\[
\psi (y) = \hsnorm{(\tilde{f}*f) (y)}^2 = \norm[\big]{\Ex{x\sim G}{f(x)^*f(xy)}}_{\mathrm{HS}}^2 \;.
\]
Therefore, if $\psi$ has small algebra norm then, $\norm{f}_{\Ut}$ can be approximated by averaging $\psi$ over an $\epsilon$-biased set. We prove the algebra norm bound by proving that the function, $\psi$, is a positive-definite function.

\begin{lemma}\label{lem:normispd}
For $f: G \to \matr M_t(\C)$, the function $\psi (y) = \hsnorm{(\tilde{f}*f) (y)}^2$ is positive-definite.
\end{lemma}
\begin{proof}
	To prove that $T_\psi$ is a PSD matrix, we wish to show that for any $\{c_a\}_{a\in G} \in \C^G$,  
	\[
	\sum_{a,b \in G} c_a \overline{c_b}\, \psi(a^{-1}b) ~\geq~ 0
	\]
The key observation is that, if we can write $\psi(a^{-1}b)  = \Ex{x,y\sim G}{ \ip{N_{x,y}(a)}{ N_{x,y}(b)}_\tr} $ for some $N_{x,y}$, then $\psi$ is positive-definite. 
We first show such a factorization of $\psi(a^{-1}b)$.

\begin{align*}
	\psi(a^{-1}b)	 ~&=~  \norm[\Big]{\Ex{x\sim G}{f(x)^*f(x a^{-1}b)}}^2_\HS\\
	~&=~    \norm[\Big]{\Ex{x\sim G}{f(xa)^*f(xb)}}^2_\HS\\
	~&=~  \ip{\Ex{x\sim G}{f(xa)^*f(xb)}}{\Ex{y\sim G}{f(ya)^*f(yb)}}_\tr \\
	~&=~  \Ex{x,y\sim G}{ \ip{f(xa)^*f(xb)}{f(ya)^*f(yb)}_\tr}  \\
	~&=~  \Ex{x,y\sim G}{ \ip{f(ya)f(xa)^*}{f(yb)f(xb)^*}_\tr} \\
	~&:=~  \Ex{x,y\sim G}{ \ip{N_{x,y}(a)}{ N_{x,y}(b)}_\tr}  
	\end{align*}
	The second last equality uses the cyclicity of trace. The result now follows,
\begin{align*}
	\sum_{a,b \in G} c_a \overline{c_b}\, \psi(a^{-1}b) 
	~&=~ \sum_{a,b \in G} c_a \overline{c_b}\Ex{x,y\sim G}{ \ip{N_{x,y}(a)}{N_{x,y}(b)}_\tr} \\
	~&=~ \Ex{x,y\sim G}{ \ip{ \sum_{a\sim G} c_a N_{x,y}(a)}{\sum_{b\sim G} c_b N_{x,y}(b)}_\tr} \\
	~&=~ \Ex{x,y \sim G}{\norm[\Big]{\sum_{a\sim G} c_a N_{x,y}(a)}_{\mathrm{HS}}^2} \geq 0.
\end{align*}
\end{proof}
\begin{remark}
	One can also deduce this by coupling \textit{Stinespring's dilation theorem} with the observation in~\cite{DOT18} that $\tilde{f}*f$ is \textit{completely positive}. This immediately yields that $\psi(y) = \ip{VV^*}{ \pi(y) VV^* \pi(y)^*} =\ip{VV^*}{ \rho(y) VV^*}$ for some representation $\rho$. We will use this idea in proving~\cref{lem:matrix_conv} which is a general version of the above lemma.  
\end{remark}

\section{Derandomized Matrix Correlation Testing}
In this section, we will focus on functions of the form, $f: G\rightarrow \Utt$. Let $S\subseteq G$ be an $\ep$-biased set. 
We consider the following robust variant of the BLR test on group $G$ :

\fbox{\begin{minipage}{12 cm}
		~~
		\\
		{\bf BLR}$_{\gamma}(G,S,f)$:
		\begin{enumerate}
			\item Sample $x\sim G, y\sim S$. 
			\item If $\norm{f(xy)-f(x)\cdot f(y)}^2_{\HS} ~\leq~ \gamma t$, output {\it Pass}. Else, output {\it Fail}. 
			~~
			\\
		\end{enumerate}
		
\end{minipage}}
\\
\\
If $S=G$, i.e., in the full randomness regime,  it can be easily shown that if a function passes the test, then it must have a large $\Ut$-norm. Our key technical claim (\cref{clm:testPassing}) is that if $S$ is a small--biased set, then, essentially, the same conclusion can be drawn from derandomized BLR test passing. 

This lower bound on the $\Ut$-norm can then be plugged into the result of Gowers and Hatami~\cite{GH17}, who showed that if a matrix-valued function on a finite group has non-trivial $\Ut$-norm then it must be close to some genuine representation. More specifically,
\begin{theorem}[Gowers--Hatami~\cite{GH17}]
	\label{thm:GowersHatami}
Let $G$ be any finite group and let $f: G\rightarrow \matr M_t(\C)$ be a matrix-valued function such that $\opnorm{f(x)}\leq 1$ and $\norm{f}_{\Ut}^4\geq ct$, for some $c >0$. Then there are $t'\in [\frac{c}{2-c}t, \frac{2-c}{c}t]$ and a function $g(x):=V\pi(x)U^*$ where $\pi$ is a $t'$ dimensional unitary representation, $U,V$ are $t\times t'$ dimensional  partial unitary matrices, such that:
	\[\E_{x\sim G}\brackets[\Big]{\angles[\big]{f(x),g(x)}_{\HS}} ~\geq~ c^2/4\]
\end{theorem} 

We first prove the key derandomization claim which lets us move from the test passing probability of the derandomized test, to a claim about the $\Ut$ norm over the entire group.

\begin{claim}[Derandomized Test also implies large $\Ut$-norm]
	\label{clm:testPassing}
	Let $\gamma,\delta\geq 0$ and $f: G\rightarrow \Utt$. If $f$ passes the $\text{BLR}_{2\gamma}(G,S,f)$ test with probability $\geq \delta$ then,
	  \[ \norm{f}_{\Ut}^4 ~\geq~ \parens[\big]{(\delta-\gamma)^2 -\epsilon}\cdot t.\]
	\end{claim}
\begin{proof}
	Let, $\Delta(x,y):=\hsnorm{f(x)f(y)-f(xy)}^2$ and $\delta'=\delta-\gamma$. We have,
	\begin{align}
		\E_{x\sim G, y\sim S}\brackets[\big]{\Delta(x,y)}
		~&=~ 2t - \, \E_{y\sim S}\brackets[\Big]{\, \ip{f(y)}{\tilde{f}*f(y)}_\tr + \ip{\tilde{f}*f(y)}{f(y)}_\tr } \label{eqn:expression}
	\end{align}
This follows directly by expanding $\Delta(x,y)$ and using the fact that $\norm{f}^2=t$. On the other hand, from the test-passing guarantee we have, 
	\begin{align}
\nonumber		\Ex{x\sim G, y\sim S}{\Delta(x,y)}
		~&\leq~  \Pr{x\sim G, y\sim S}{\Delta(x,y)>2\gamma t}\cdot 2t+
	\Pr{x\sim G, y\sim S}{\Delta(x,y)\leq  2\gamma t}\nonumber \cdot 2\gamma t\\
\nonumber		~&~\leq 2t(1-\delta)+2\gamma t\\
		~&=~ 2(1-\delta')t \label{eqn:upper_bound}
	\end{align}
	In the first inequality, we used the fact that $\max_{x,y}\{\Delta(x,y)\}\leq2t$, and in the second inequality, we used test passing probability to upper bound  $\Pr{x\sim G, y\sim S}{\Delta(x,y)>2\gamma t}$. Combining~\cref{eqn:expression} and~\cref{eqn:upper_bound}, we get:
	\begin{align*}
		2\delta' t
		~&\leq~ \E_{y\sim S}\brackets[\Big]{\, \ip{f(y)}{\tilde{f}*f(y)}_\tr + \ip{\tilde{f}*f(y)}{f(y)}_\tr } 
				\\~&\leq~ 2\Ex{y\sim S}{\hsnorm{f(y)}\cdot \hsnorm{\tilde{f}*f(y)}}
			&&(\text{\CS})
			\\~&=~ 2\sqrt{t} \cdot\Ex{y\sim S}{\hsnorm{\tilde{f}*f(y)}} 
		&&(\text{Using: $f$ is unitary-valued.})
		\\~&\leq~ 2\sqrt{t} \cdot \parens[\Big]{\Ex{y\sim S}{\hsnorm{\tilde{f}*f(y)}^2}}^{\frac{1}{2}} &&(\text{\CS})
	\end{align*} 
Now we define $\psi(y)= \hsnorm{\tilde{f}*f(y)}^2$. Observe that, $\norm{\psi}_{\infty}\leq t$ and it is a positive-definite function by~\cref{lem:normispd}. This will allow us to deduce that the $\Ut$-norm is large easily. From the computation before, we get,
\begin{align*}
\delta'^2 t ~&\leq~ \Ex{y\sim S}{\psi(y)} \\
~&\leq~\Ex{y\sim G}{\psi(y)} + \ep  \norm{\psi}_A
&&(\text{By Lemma~\ref{lem:pd_fooled}})
\\~&\leq~  \Ex{y\sim G}{\psi(y)} + \ep t &&(\text{By Lemma~\ref{lem:normispd} and Corollary~\ref{cor: pdFooled}})
\\~&=~ \Ex{y\sim G}{ \hsnorm{\tilde{f}*f(y)}^2} + \ep t
\\
(\delta'^2-\ep)\, t ~&\leq~ \norm{f}_{\Ut}^4.  &&(\text{By definition of $\Ut$-norm}) \qedhere
\end{align*}
\end{proof}

To prove our main result, we need one more component that roughly says that the convolution of functions is mostly supported on low dimensional irreps. The proof is almost identical to the proof of the BNP lemma by Babai, Nikolov, and Pyber~\cite{BNP08}. 

\begin{claim}
	\label{clm:MxBNP}
	Let, $f,g: G\to \matr M_t(\C) $ and let $T:=\{\rho \in \widehat{G} : d_\rho \geq D\}$, then the following holds: \[\sum\limits_{\rho\in T }d_{\rho}\norm[\big]{\widehat{f*g}}_{\mathrm{HS}}^2 ~\leq~ \frac{1}{D}\,\norm{f}^2_2\,\norm{g}^2_2.\] In particular, if  $G$ is a $D$-quasirandom group, then $\norm[\big]{f*g-\mu(f*g)} \leq \frac{1}{\sqrt{D}}\,\norm{f}_2\,\norm{g}_2.$
\end{claim}
\begin{proof}
	From the convolution identity~(\cref{fact:fourier}), we have: 
	\begin{align*}
		\sum\limits_{\rho\in T }d_{\rho}\norm[\big]{\widehat{f*g}}_{\mathrm{HS}}^2~&=~ 
		\sum\limits_{\rho \in T }d_{\rho}\norm[\big]{\widehat{f}(\rho)\widehat{g}(\rho)}_{\mathrm{HS}}^2 &&(\text{By convolution identity})
		\\~&\leq~\sum\limits_{\rho \in T }d_{\rho}\norm[\big]{\widehat{f}(\rho)}_{\mathrm{HS}}^2\norm[\big]{\widehat{g}(\rho)}_{\mathrm{HS}}^2&&(\text{Norm submultiplicativity})
	   \\~&\leq~\frac{1}{D}\sum\limits_{\rho \in T }d_{\rho}^2\norm[\big]{\widehat{f}(\rho)}_{\mathrm{HS}}^2\norm[\big]{\widehat{g}(\rho)}_{\mathrm{HS}}^2
	   &&(\text{Using $d_\rho \geq  D$ for $\rho\in T$})
	    \\~&\leq~\frac{1}{D} \sum\limits_{\rho}d_{\rho}\norm[\big]{\widehat{f}(\rho)}_{\mathrm{HS}}^2 \sum\limits_{\rho}d_{\rho}\norm[\big]{\widehat{g}(\rho)}_{\mathrm{HS}}^2
	    \\~&=~\frac{1}{D}\norm{f}^2_2\norm{g}^2_2&&(\text{Parseval's identity})
	\end{align*}
To see the final claim, apply the above to $T:=\{\rho \in \widehat{G} : d_\rho \geq D\}=\{\rho \neq \triv\}$ because the group $G$ is $D$-quasirandom. Using Parseval's identity~(\cref{fact:fourier}) we obtain,
\begin{align*}
	\norm{f*g-\mu(f*g)}^2= \sum\limits_{\rho\neq  \triv }d_{\rho}\norm[\big]{\widehat{f*g}}_{\mathrm{HS}}^2\leq  \frac{1}{D}\norm{f}^2_2\norm{g}^2_2.\qedhere
\end{align*}
\end{proof}

\begin{restatable}[Deranomized Homomorphism Testing]{theorem}{main}\label{thm:main}
Let $G$ be any finite group, and $f: G\rightarrow \Utt$ be a unitary matrix-valued function. Let $S\subseteq G$ be an $\ep$-biased set. Assume that $f$ passes the $\text{BLR}_{2\gamma}(G,S,f)$ test with probability $\geq \delta (\gamma)$ for any chosen $0\leq \gamma\leq 1$. Then for $\eta= (\delta-\gamma)^2-\ep$, the following holds,
	\begin{enumerate}
		\item There exists $t'\in [\eta t, \frac{2}{\eta}t]$ and a function $g(x):=V\pi(x)U^*$ where $\pi$ is a $t'$ dimensional unitary representation, $U,V$ are $t\times t'$ dimensional  partial unitary matrices, such that:
				\[\E_{x}{\angles[\big]{f(x),g(x)}_{\HS}} \geq \frac{\eta^2}{4} t \]
		\item for any integer $D>1$,
		\[ \max_{\rho \in \widehat{G}: d_\rho < D} \|\widehat{f}(\rho)\|^2_{\HS} ~\geq~ \eta-\frac{t}{D}  \]
		In particular, there exists an irrep $\rho$ such that $d_{\rho} < \frac{2t}{\eta}$ such that $\|\widehat{f}(\rho)\|^2_{\HS} \geq \frac{\eta}{2}$. 
	\end{enumerate}
	\end{restatable}

\begin{proof}
	From the test passing assumption and~\cref{clm:testPassing}, it follows that: $\norm{f}_{\Ut}^4\geq \eta t$.	
	 						Now applying,~\cref{thm:GowersHatami} gives us the first claim. For the second claim, our starting point is the same: 
\begin{align*}
	\eta t\leq \norm{f}_{\Ut}^4
	=\norm{\tilde{f}*f}^2=\sum_{\rho \in  \widehat{G} }d_{\rho}\norm[\big]{\widehat{\tilde{f}*f} \,(\rho)  }_{\mathrm{HS}}^2
\end{align*}
Now we divide $\widehat{G}$ into low and high dimensional irreps by taking $T:=\{\rho \in \widehat{G}:~d_\rho\geq D \}$.  We have,
\begin{align*}
	\eta t\leq \sum_{\rho \in  \widehat{G} }d_{\rho}\norm[\big]{\widehat{ \tilde{f}*f}(\rho)}_{\mathrm{HS}}^2 ~&=~ \sum_{\rho \in  \widehat{G}\setminus T }d_{\rho}\norm[\big]{\widehat{\tilde{f}*f}\,(\rho)}_{\mathrm{HS}}^2 + \sum_{\rho \in  T }d_{\rho}\norm[\big]{\widehat{\tilde{f}*f} \,(\rho)}_{\mathrm{HS}}^2 
	\\~&\leq~ \sum_{\rho \in  \widehat{G}\setminus T }d_{\rho}\norm[\big]{\widehat{\tilde{f}*f} \, (\rho)}_{\mathrm{HS}}^2 + \frac{\norm{\tilde{f}}^2\norm{f}^2}{D}.
\end{align*}
In the inequality step above, we used~\cref{clm:MxBNP}. As $f$ is unitary valued, $\norm{f}^2=\norm{\tilde{f}}^2=t$. It follows that:  $\sum_{\rho \in  \widehat{G}\setminus T }d_{\rho}\norm[\big]{\widehat{\tilde{f}*f}(\rho)}_{\mathrm{HS}}^2 \geq \eta t - t^2/D$. Finally, we have,
\begin{align*}
	\eta t - t^2/D ~&~\leq \sum_{\rho \in  \widehat{G}\setminus T }d_{\rho}\norm[\big]{\widehat{\tilde{f}*f} \,(\rho)}_{\mathrm{HS}}^2\\~&~=\sum_{\rho \in  \widehat{G}\setminus T }d_{\rho}\norm[\big]{\widehat{f}(\rho)^*\widehat{f}(\rho)}_{{\mathrm{HS}}}^2&&(\text{Convolution identity,~\cref{fact:fourier}})
	\\~&\leq~\sum_{\rho \in  \widehat{G}\setminus T }d_\rho \norm{\widehat{f}(\rho)}_{{\mathrm{HS}}}^4&&(\text{Sub-Multiplicativity.})
	\\~&\leq~ \max_{\rho \in \widehat{G}: d_\rho < D} \|\widehat{f}(\rho)\|^2_{\HS} \cdot \sum_{\rho }d_\rho \norm{\widehat{f}(\rho)}_{{\mathrm{HS}}}^2&&(\text{As, $d_\rho\leq D$ holds for any $\rho \in \widehat{G}\setminus T$})
	\\~&=~ \max_{\rho \in \widehat{G}: d_\rho < D} \|\widehat{f}(\rho)\|^2_{\HS} \cdot \norm{f}^2&&(\text{Parseval's identity,~\cref{fact:fourier}})
	\\~&=~ t\cdot \max_{\rho \in \widehat{G}: d_\rho < D} \|\widehat{f}(\rho)\|^2_{\HS}\mper&&(f\text{ is unitary, } \norm{f}^2 = t)\qedhere
\end{align*}
\end{proof}
\section{Derandomized Mixing}
The goal of this section is to prove a general \enquote{degree--2 mixing lemma} as explained in the introduction for the general case of matrix-valued functions. The assumption of the functions being mean-zero is without loss of generality and only for the sake of brevity.

\bnp*

To prove this derandomization, we first prove a more general version of~\cref{lem:normispd}, where instead of $\psi(y) = \hsnorm{(\tilde{f}*f)(y)}^2$, we have $\psi(y) = \hsnorm{(f*g)(y)}^2$. This is no longer positive definite as earlier, but we can use elementary representation theory to explicitly give a factorization of the form $\psi(y) = \ip{u}{\phi(y) \, v}$ and thereby compute the algebra norm. The representation $\phi$ that will come up is defined as follows --- let $V \subseteq L^2_t(G \times G)$ be the subspace  $V= \mathrm{span} \braces[\big]{F(x,y) =  f(y)f(x)^* \,\vert\, f\in L^2_t(G) }$. Note that $V$ inherits the expectation trace inner product, \ie $\ip{F}{H} = \Ex{x,y\sim G}{\ip{F(x,y)}{H(x,y)}_\tr}$. Then we can define the following representation of $G$, \[(\phi(a) \cdot F) (x,y) = F(xa,ya) = f(ya)f(xa)^* .\]
We are now ready to prove the generalization of~\cref{lem:normispd}.
\begin{lemma}\label{lem:matrix_conv}
	Let $f, h : G \to \matr M_t(\C)$  and define $  \psi(y) = \hsnorm{(h*f)\,(y)}^2$. Then, $\norm{\psi}_A \leq \norm{f}^2\norm{h}^2$. Morever, if the functions are unitary-valued, then $\norm{\psi}_A \leq t$. 
\end{lemma}
\begin{proof}
Let $F(x,y)= f(y)f(x)^*$ and similarly, $\tilde{H}(x,y) = \tilde{h}(y)\tilde{h}(x)^*$. 
Now,
	\begin{align*}
	\psi(a)	 ~&=~  \norm[\Big]{\Ex{x\sim G}{h(x^{-1})f(x a)}}^2\\
	~&=~  \ip{\Ex{x\sim G}{\tilde{h}(x)^*f(xa)}}{\Ex{y\sim G}{\tilde{h}(y)^*f(ya)}}_\tr \\
	~&=~  \Ex{x,y\sim G}{ \ip{\tilde{h}(x)^*f(xa)}{\tilde{h}(y)^*f(ya)}_\tr}  \\
	~&=~  \Ex{x,y\sim G}{ \ip{\tilde{h}(y)\tilde{h}(x)^*}{f(ya)f(xa)^*}_\tr}  \\
	~&=~  \Ex{x,y\sim G}{ \ip{\widetilde{H}(x,y)}{\phi(a)\, F(x,y)}_\tr} \\
	~&=~  \ip{\widetilde{H}}{\phi(a)\, F} 
	\end{align*}
	
	By the definition of algebra norm~\cref{def:alg_norm}, we have $ \norm{\psi}_A \leq \norm{H}\cdot \norm{F}$. Now,	
 \[
\norm{F}^2 = \Ex{x,y}{\hsnorm{f(y)f(x)^*}^2} \leq \Ex{x,y}{\hsnorm{f(y)}^2\hsnorm{f(x)^*}^2} = \norm{f}^4.
\]
The inequality here follows from sub-multiplicativity. Similarly, $\norm{\widetilde{H}} = \norm{h}^2$. This proves the first claim. When the functions map to unitary matrices, $\widetilde{H}(x,y), F(x,y)$ are unitary and thus, $\norm[\big]{\widetilde{H}(x,y)}_{\HS}^2 = \hsnorm{F(x,y)}^2 = t$ for every $x, y \in G$. We can thus avoid using sub-multiplicativity and directly obtain, 	

 \[
\norm{F}^2 = \Ex{x,y}{\hsnorm{f(y)f(x)^*}^2} = t = \norm{\widetilde{H}}^2 .\qedhere
\]

\begin{remark}
One can weaken the unitary assumption by requiring that $f$ is \enquote{unitary on average}, an assumption also used in~\cite{MR15}. This is because $\norm{F} = \hsnorm{\Ex{x}{f(x)f(x)^*}} $.
\end{remark}
 \end{proof}

\bnp*

 \begin{proof}
From~\cref{lem:matrix_conv}, we have that the function, $\psi(s) := \hsnorm{(f*g)\,(s)}^2$ has algebra norm $\norm{f}_2^2\,\norm{g}_2^2$ and therefore from~\cref{lem:pd_fooled} we have that averaging over $S$ is $\eps$-close to true average.
Thus, 
\begin{align*}
	 \Ex{s\sim S}{\hsnorm{ (f*g)(s)}^2} ~&\leq~ 
	 \Ex{s\sim G}{\hsnorm{ (f*g)(s)}^2}  +\eps \norm{f}_2^2\,\norm{g}_2^2  && \text{(Using~\cref{lem:pd_fooled})}
	 \\~&\leq~ \frac{1}{D} \norm{f}_2^2\,\norm{g}_2^2 +	 \eps\norm{f}_2^2\,\norm{g}_2^2 && \text{(Using~\cref{clm:MxBNP})} \qedhere
\end{align*}
 \end{proof}

 \section*{Acknowledgements}
\vspace{-0.3cm}
We are grateful to Fernando Granha Jeronimo for several illuminating conversations in the early phase of this work and to Madhur Tulsiani for closely looking at the draft and providing several helpful suggestions.

\bibliographystyle{alphaurl}
\bibliography{macros,references}

 \end{document}